\documentclass[final,twoside,11pt]{entics} 
\usepackage{enticsmacro}
\usepackage{graphicx}
\usepackage[all]{xy}
\usepackage{mathtools} %declaremathoperator
\usepackage{amssymb}
\usepackage{stmaryrd} % denotation brackets
\usepackage{mathrsfs} % mathscr
\usepackage{mathpartir} % inference rules
\usepackage{subcaption} % subfigures
\usepackage{empheq} % math frames
\usepackage{cancel}
\usepackage{./defs}

\def\conf{MFPS 2023} 	%%Fill in the acronym for your conference (with year)
\def\myconf{mfps39}
\volume{3}			%Fill in the ENTICS volume number here
\def\volu{3}			% and here
\def\papno{19}			%Fill in your paper number here
\def\lastname{Somayyajula and Pfenning} %Lastnames appear in the running header 
\def\runtitle{Dependent Type Refinements for Futures} %Short title appears in the running header on even pages. 
\def\mydoi{12286}
\def\copynames{S. Somayyajula, F. Pfenning}    %% Fill in the first initial and last name of the authors
\begin{document}
%%%Note the beginning and end of the frontmatter section that starts here%%%%%
\begin{frontmatter}
  \title{Dependent Type Refinements for Futures} 						%%Title here and the\thanksref{ALL}
 %\thanks[ALL]{}   %%Text of \thanks[ALL} here..
 %%%%%%%%%%%%%%%%%%%%%%%%%%%%			This Thanks is optional.
  %%%%Now the author(s) names(s)%%%%%
  \author{Siva Somayyajula\thanksref{myemail}}	%%Note NO SPACE between 
   \author{Frank Pfenning\thanksref{coemail}}		%last name and \thanksref{...} 
    %%%Next come the addresses%%%%
   \address{Computer Science Department\\ Carnegie Mellon University\\				%or between \thanksrefs...
    Pittsburgh, USA}  							
   \thanks[myemail]{Email: \href{mailto:ssomayya@cs.cmu.edu} {\texttt{\normalshape
        ssomayya@cs.cmu.edu}}} 
   %%%Note: if both authors share same institution, only list the address once, after the second 
   %%%author. 
   %%%There also is a link from the first author to the co-author's address to show how to list 
   %%%affiliations to more than one institution, when needed. 
  \thanks[coemail]{Email:  \href{mailto:fp@cs.cmu.edu} {\texttt{\normalshape
        fp@cs.cmu.edu}}}
\begin{abstract} 
Type refinements combine the compositionality of typechecking with the expressivity of program logics, offering a synergistic approach to program verification. In this paper we apply dependent type refinements to SAX, a futures-based process calculus that arises from the Curry-Howard interpretation of the intuitionistic semi-axiomatic sequent calculus and includes unrestricted recursion both at the level of types and processes. With our type refinement system, we can reason about the partial correctness of SAX programs, complementing prior work on sized type refinements that supports reasoning about termination. Our design regime synthesizes the infinitary proof theory of SAX with that of bidirectional typing and Hoare logic, deriving some standard reasoning principles for data and (co)recursion while enabling information hiding for codata. We prove syntactic type soundness, which entails a notion of partial correctness that respects codata encapsulation. We illustrate our language through a few simple examples.
\end{abstract}
\end{frontmatter}

\section{Introduction}
Type refinements internalize assertions into the type structure of functional programs, combining the compositionality of typechecking with the expressivity of program logics. While having comparable verification capabilities to ``traditional'' dependent type theories \cite{Vazou2018popl}, \emph{dependent type refinements} can also have their assertion logics extended for domain-specific verification \cite{Xi1999popl}. This line of work has also revealed close connections between the proof theory of Hoare logic and that of static type refinement discipline---some correspondences include path-sensitive elimination rules to the conditional rule, substitution to composition, and subsumption to consequence.

Recent work on \emph{dependent session types} \cite{Toninho2011ppdp} have begun to transport related results to process calculi, yet are limited by the need to carefully interface the linearity inherent in session types with type dependency. Solutions tending towards traditional type theory marry a separate dependently-typed proof language with a session-typed process calculus \cite{Toninho2018fossacs}, whereas those in the space of type refinements are typically limited to an index language incapable of expressing the complete nuances of process dynamics \cite{Toninho2021ppdp}. This raises the question: is it possible to stake out a middle ground and adapt expressive dependent type refinements to a process calculus?
In this article, we develop such type refinements within SAX, a futures-based process calculus that arises from the Curry-Howard interpretation of the intuitionistic semi-axiomatic sequent calculus \cite{DeYoung2020fscd}. Thus, we are not bound by the constraints of linearity. However, two major questions generate the design space:
\begin{enumerate}
\item What is a suitable program logic and to what extent should/can it account for process dynamics, including higher-order (co)data?
\item How do types internalize assertions and how should the type system be presented?
\end{enumerate}
A core desideratum typically forces certain answers to these questions---for example, to ensure decidable typechecking, \emph{liquid types} \cite{Rondon2008pldi} are designed around the following answers:
\begin{enumerate}
\item Assertions are quantifier-free, with quantifiers encoded as dependent types \cite[Section 3]{Vazou2018popl}. Functions and their extensional equality are encoded by first-order approximation \cite[Sections 4 and 5]{Vazou2018popl} and typeclasses \cite{Vazou2022haskell}, respectively, in contrast to higher-order program logics \cite{Gianas2008mpc}.
\item Refined types are distinguished type constructors and typing is \emph{bidirectional} \cite{Pierce2000acm,Dunfield2022acm}.
\end{enumerate}
Our guiding principle is to follow the proof theory of the semi-axiomatic sequent calculus, answering these questions as follows:
\begin{enumerate}
\item We develop a first-order theory of SAX values, function applications, and lazy record projections. In particular, we \emph{avoid} directly reifying processes into the assertion logic.
\item We use bidirectional typing not to effect algorithmic typechecking, but to determine the shape of types and their associated typing rules. In particular, we synthesize the semi-axiomatic sequent calculus with Hoare logic, viewing the former as an intermediate point between (bidirectional) natural deduction and the (unidirectional) sequent calculus. While we recover the expected properties for positive (data) types, curiously, refinements nested under a negative (codata) type may hide information from assertions attached to said type, providing a facility for codata \emph{encapsulation}.

Considering recursion both at the level of types and processes, our type system establishes partial correctness, complementing \emph{sized type refinements} used to guarantee termination in SAX \cite{Somayyajula2022fscd}. Following \emph{op.~cit.}, we view (sub)typing derivations for equirecursive types \cite{Lakhani2022esop} and recursive programs, respectively, as \emph{infinite proofs} \cite{Brotherston,Derakhshan19arxiv} generated by \emph{mixed inductive-coinductive} inference systems \cite{Basold2018,Danielsson2010mpc}. Standard \emph{assume-guarantee reasoning} for recursive programs that arises from typing derivation circularity, when combined with subsumption, uniformly admits reasoning with induction, coinduction, and mixed induction and coinduction within the language. %The inclusion of quantifiers into the assertion logic as well as our use of infinite proofs jeopardize decidability of typechecking, which we discuss in the conclusion.
\end{enumerate}
In summary, our primary contribution is Dependent Refined SAX (DRSAX): a dependent type refinement system for SAX (Section \ref{sec:drsax}) with a corresponding type soundness result entailing \emph{observable} partial correctness (Section \ref{sec:semantics}). Essentially, we show that assertions about objects that are not encapsulated hold directly. Our secondary contribution is the design regime listed above, which leads us to a novel derivation of codata encapsulation and a uniform consideration of induction, coinduction, and mixed induction and coinduction within DRSAX.
\section{DRSAX: Semi-Axiomatic Type Theory Meets Hoare Logic}\label{sec:drsax}
In this section, we develop DRSAX by first commenting on the judgmental structure of the semi-axiomatic sequent extended for our purposes. After quickly reviewing some ancillary definitions, we examine the relevant typing rules with examples.
\subsection{Judgmental Structure}\label{sec:drsax-judgments}
In short, the semi-axiomatic sequent calculus replaces the typical right and left rules for positive and negative types (i.e., the \emph{non-invertible} ones), respectively, with axioms. The corresponding typing judgment for processes takes on the following form:
\[\overbrace{x:A,\ldots,y:B}^{\displaystyle\Gamma}\vdash P(x,\ldots,y,z)\div(z:C)\]
The \emph{process} $P$ may perform blocking reads from \emph{source addresses} $x,\ldots,y$ of \emph{futures} and must perform a non-blocking write exactly once to the future addressed by the \emph{destination} $z$ according to \emph{types} $A,\ldots,B$ and $C$, respectively, corresponding to asynchronous communication with futures \cite{Halstead85}. Data addressed are values and process continuations of \emph{positive} and \emph{negative} type, respectively, recalling the binary term and type distinction of call-by-push-value \cite{Levy1999tlca}. To make the jump to Hoare logic, we attach \emph{preconditions} to the antecedents and a \emph{postcondition} to the succedent with \emph{refined types}: the judgment becomes $x:\A$ where $\A$ has the form $A\mid\lambda x.\,\phi(x)$. $\phi,\psi,\chi,\ldots$ are \emph{assertions} from the external logical theory described in the next sub-section. As shorthand, we write $x:A\mid\phi(x)$ to conflate the antecedent or succedent variable with that bound in $\phi$ or $\phi(\cdot)$ when the $x$ need not be mentioned.
\[x:\A,\ldots,y:\B(x,\ldots)\vdash P(x,\ldots,y,z)\div(z:\CC(x,\ldots,y))\]
We encounter questions immediately---for example, considering disjunction as a labelled sum type $\oplus\{\ell:A_\ell\}_{\ell\in S}$, which of the following is a more appropriate right axiom? Note that any non-refined type $A$ is canonically the refined type $A\mid\top$.
\[\infer[{$\oplus$}R, {$k\in S$}]{}{\Gamma,x:A_k\vdash\ldots\div(y:\oplus\{\ell:A_\ell\}_{\ell\in S}\mid y\equiv k\cdot x)}\quad\text{vs.}\quad\infer[{$\oplus$R, $k\in S$}]{}{\Gamma,x:A_k\mid\phi(k\cdot x)\vdash\ldots\div(y:\oplus\{\ell:A_\ell\}_{\ell\in S}\mid\phi(y))}\]
The framework for bidirectional typing in \cite{Dunfield2004popl,Dunfield2022acm} tells us which: the principal judgment's data in a positive introduction rule are always \emph{inputs}. Thus, the second rule is clearly canonical, as $\oplus\{\ell:A_\ell\}_{\ell\in S}$ and $\phi(y)$ are inputs unlike in the first rule. Peculiarly, $A_k$ and $\phi(k\cdot x)$ are \emph{outputs}, indicating that type information flows from right to left, like in \emph{backwards bidirectional typing} \cite{Chlipala2005tldi,Zeilberger2015types,Dunfield2022acm}. In particular, this rule simulates bottom-up flow in natural deduction:
\[\infer[{$\oplus$I, $k\in S$}]{\Gamma\vdash\ldots:A_k}{\Gamma\vdash\ldots:\oplus\{\ell:A_\ell\}_{\ell\in S}}\]
Dually, left axioms for negative types flow from left to right, corresponding to the top-to-bottom flow for negative eliminations in natural deduction. Consider the following example of negative conjunction as a lazy record type (omitting refinements in the axiom for the moment).
\[\infer[{\&E, $k\in S$}]{\Gamma\vdash\ldots:\&\{\ell:A_\ell\}_{\ell\in S}}{\Gamma\vdash\ldots:A_k}\quad\rightsquigarrow\quad\infer[{\&L, $k\in S$}]{}{\Gamma,y:\&\{\ell:A_\ell\}_{\ell\in S}\vdash\ldots\div(x:A_k)}\]
Taking a step back and thinking of a (non-semi-axiomatic) sequent calculus as an algorithmic type system, the resulting hypothetical judgment would appear as follows: 
\[x\To A,\ldots,y\To B\vdash P(x,\ldots,y,z)\div(z\From C)\]
where the arrows $\To$ and $\From$ distinguish between antecedent and succedent judgments, respectively, all of which are inputs. Thus, we would have \emph{unidirectional} process typing from the bottom up. As we have seen however, the semi-axiomatic sequent inherits some bidirectionality from natural deduction, in which both antecedents and the succedent may be outputs. Thus, we must allow $\From$ and $\To$ judgments to respectively appear in $\Gamma$ and the succedent as outputs.  We call the resulting hypothetical judgment(s) \emph{process typing}, defined in Figure \ref{fig:statics-judg} along with ancillary judgments (wellformedness judgments are standard and omitted).
\begin{figure}
\begin{empheq}[box=\fbox]{align*}
&\text{judgments }&&J:=(x\From A\mid\phi(x))\mid(x\To B\mid\phi(x))\\
&\text{contexts }&&\Gamma:=\cdot\mid\Gamma,J\\
&\text{assertion sequent}&&\Gamma\vdash\phi\\
&\text{subtyping}&&\Gamma\vdash A\leq B\\
&\text{process typing}&&\Gamma\vdash P\div J
\end{empheq}
\caption{Judgments}
\label{fig:statics-judg}
\end{figure}
\begin{figure}
\centering
\HRule\\
\begin{subfigure}[b]{0.56\textwidth}
\begin{tabular}{r l l}
$A:=$   & $A^+\mid A^-$\\
$\mid$  & $X$                                   & recursive type ($X=A$)\\
$\A:=$ & $(A\mid\lambda x.\,\phi(x))$ & refined type\\
$A^+:=$ & $\mathbf{1}$                          & (positive) unit\\
$\mid$  & $(x:\A)\otimes B(x)$      & dependent eager pairs\\
$\mid$  & ${\oplus}\{\ell:A_\ell\}_{\ell\in S}$ & eager sums\\
$A^-:=$ & $(x:\A)\to\B(x)$          & dependent functions\\
$\mid$  & ${\&}\{\ell:\A_\ell\}_{\ell\in S}$     & lazy records
\end{tabular}
\caption{Types}
\end{subfigure}
\vline
\begin{subfigure}[b]{0.43\textwidth}
\begin{tabular}{r l l}
$s,t:=$   & $x,y,z,\ldots$ & address variables\\
$\mid$  & $a,b,c,\ldots$ & runtime addresses
\end{tabular}
\begin{tabular}{r l l}
$V:=$  & $\pair{}$ & ($\unit\R$)\\
$\mid$ & $\pair{s,t}$ & (${\otimes}\R,{\to}\L$)\\
$\mid$ & $\ell\cdot t$ & (${\oplus}\R,{\&}\L$)
\end{tabular}\\\\
\begin{tabular}{r l l}
$K:=$  & $\pair{}\To P$ & ($\unit\L$)\\
$\mid$ & $\pair{x,y}\To P(x,y)$ & (${\otimes}\L,{\to}\R$)\\
$\mid$ & $\{\ell\cdot x\To P_\ell(x)\}_{\ell\in S}$ & (${\oplus}\L,{\&}\R$)
\end{tabular}
\caption{Values and Continuations}
\end{subfigure}
\HRule\\
\begin{subfigure}[c]{\textwidth}
\begin{tabular}{r l l l}
$P,Q:=$ & $t\from s$         & copy contents of $s$ to $t$ & (identity)\\
$\mid$  & $x\from P(x);Q(x)$ & spawn $P$ writing to $x$, & (cut)\\
        &                    & proceed concurrently as $Q$\\
$\mid$  & $t.V$              & write $V$ to $t$, or & (positive right/negative left rule)\\
        &                    & pass $V$ to continuation in $t$\\
$\mid$  & $\Case t\,K$       & pass value in $t$ to $K$ or, & (negative right/positive left rule)\\
        &                    & write $K$ to $t$\\
$\mid$ & $(x:\A)\anno P$ & refined type annotation & (\textsc{annoL/R})\\
$\mid$ & $f(\overline{s},t)$ & definition call & $(f(\overline{x:\A},y:\CC(\overline{x}))=P(\overline{x},y))$
\end{tabular}
\caption{Processes}
\HRule
\end{subfigure}
\caption{Syntax}
\label{fig:syntax}
\end{figure}
\subsection{Syntax}
We briefly comment on the syntax for addresses, types, and processes in Figure \ref{fig:syntax}.
\begin{itemize}
\item We distinguish between \emph{address variables} introduced in the previous subsection and \emph{runtime addresses} which only appear in Section \ref{sec:semantics}.
\item In addition to labelled sum and lazy record types, we have dependent eager pair, unit, and function types that seem asymmetrically presented---we elaborate on this later.
\item With the exception of type annotation and definition calls, which belong to an ambient signature of typed mutually recursive definitions, each process corresponds to a judgmental or logical rule in the semi-axiomatic sequent calculus. Note that $\overline{x:\A}$ is a \emph{telescope} \cite{DeBruijn1991}, i.e., a context where each subsequent binding may refer to previously bound variables.
\end{itemize}
Now, the following sub-sections elaborate on the language, starting with the assertion logic, judgmental rules, logical rules, and ending with recursive definitions.
\subsection{Assertion Logic}\label{sec:drsax-assertions}
Assertions, given by the grammar below, are drawn from the (classical) first-order theory of equality with uninterpreted functions. The ellipsis indicates the potential for extension to effect richer verification---for example, including the theory of arithmetic enables termination checking \cite{Somayyajula2022fscd}.
\begin{empheq}[box=\fbox]{align*}
\phi,\psi&:=\bot\mid\top\mid M\equiv N\mid\is_k(M)\mid\phi\land\psi\mid\phi\implies\psi\mid\,\forall x.\,\phi(x)\mid\ldots\\
M,N&:=s\mid\underbrace{\pair{}}_{\mathbf{1}\R}\mid\underbrace{\pair{M,N}}_{{\otimes}\R}\mid\underbrace{M\quo N}_{{\to}\L}\mid\underbrace{k\cdot M}_{{\oplus}\R}\mid\underbrace{M{.}k}_{{\&}\L}
\end{empheq}
We approximate process dynamics by a careful definition of first-order \emph{terms} $M,N$. First, the indirection introduced by addresses is collapsed by treating address variables as term variables and runtime addresses as nullary function symbols (\emph{not} constants, since unequal addresses do not necessarily have unequal referents). Thus, an address $s$ pointing to SAX (co)data denoted by $M$ is represented by the assertion $s\equiv M$. Finally, each axiom is assigned an uninterpreted function:
\begin{itemize}
\item\emph{Positive right axioms}: $\pair{}$ is a unit value, $\pair{M,N}$ is a pair of values $M$ and $N$, and $k\cdot M$ is a $k$-tagged value $M$. These function symbols are additionally subject to the first-order theory of \emph{non-cyclic} data structures \cite{Oppen1978popl}. Note that even with recursive types, values cannot be cyclic, because a non-allocating process cannot write to and read from the same address. For convenience, we assume the availability of the assertion $\is_k(M)$ that is true when $M$ is a $k$-tagged value.
\item\emph{Negative left axioms}: $M\quo N$ represents an application of the SAX function denoted by $M$ to argument $N$ and $M.k$ is the $k$\textsuperscript{th} projection of the record denoted by $M$. Note our use of the phrase ``the [continuation] denoted by $M$''---we do \emph{not} directly encode function bodies into the assertion logic, as that could reveal information hidden by negative type refinements discussed in Sections \ref{sec:drsax-additives} and \ref{sec:drsax-deptypes}. Instead, a continuation addressed by $s$ is abstractly described by assertions about $s\quo N$ or $M.k$---similar to \emph{copattern matching} \cite{Abel2013popl}.
\end{itemize}
\subsection{Phase Change: Subsumption and Type Annotation}
Analogous to natural deduction, changes of phase between inputs and outputs are mediated by \emph{subsumption} ($\leq$R/L) and \emph{type annotation} (\textsc{Anno}R/L). The judgmental distinction between the left- and right-hand sides of the sequent require two rules each \cite{Lengrand2006csl}.
\begin{empheq}[box=\fbox]{gather*}
\leqR\qquad\leqL\\
\annoR\qquad\annoL
\end{empheq}
Subsumption combines covariant succedent and contravariant antecedent subtyping with \emph{postcondition strengthening} and \emph{precondition weakening}, analogous to the consequence rule in Hoare logic \cite{Jhala2021}. The subtyping rules in Figure \ref{fig:subtyping-rules} generalize polarized subtyping (which includes \emph{width} and \emph{depth} subtyping for sums and records \cite{Lakhani2022esop}) to internalize type dependency \cite{Aspinall1996lics}. In particular:
\begin{itemize}
\item Introducing the auxiliary judgment $\Gamma\vdash\A\leq\B$, ${\leq}$\textsc{Pred} corresponds to the standard \emph{predicate subtyping} \cite{Rushby1998tse} rule at the core of type refinement systems \cite{Jhala2021}.
%\item As expected, ${\leq}{\otimes}$ and ${\leq}{\to}$ are covariant and contravariant in the first type component, respectively, assuming the stronger specification when inspecting the second type component \cite{Aspinall1996lics}.
%\item Sum and lazy record types enjoy \emph{width} and \emph{depth subtyping} (${\leq}{\oplus},{\leq}{\&}$) TODO CITE?, omitting \emph{emptiness} and \emph{fullness} checking \cite{Lakhani2022esop} for presentational simplicity.
\item The $\infty$ sign surrounding the premises of ${\leq}$\textsc{RecR/L} indicates a coinductive occurrence of the subtyping judgment (with all other ones being inductive) \cite{Danielsson09,Danielsson2010mpc}; \emph{ops.~cit.}~themselves build on coinductive axiomatizations of subtyping \cite{Brandt1997tlca}. That is, a subtyping derivation is a (potentially) infinitely deep tree where every infinite branch passes through an instance of this rule infinitely many times, representing the unfolding of a recursive type. See Example 6 in \cite{Lakhani2022esop} for one such derivation.
\end{itemize}
To finish, we verify that the subtyping relation is indeed reflexive and transitive via mixed induction and coinduction.
\begin{rem}[Mixed Induction and Coinduction]
Proofs by ``mixed induction and coinduction'' over the structure of (sub)typing derivations involve a lexicographic guarded coinduction to prove judgments marked $\infty$ prioritized over a structural induction on (smaller) subderivations (i.e., when guardedness does not change). Refer to \cite{Danielsson09} for further examples.
\end{rem}
\begin{lem}[Reflexivity and Transitivity of Subtyping]
\hphantom{}
\label{lem:subtyping}
\begin{itemize}
\item Reflexivity: $\Gamma\vdash A\leq A$
\item Transitivity: if $\Gamma\vdash A\leq B$ and $\Gamma\vdash B\leq C$, then $\Gamma\vdash A\leq C$
\end{itemize}
\end{lem}
\begin{proof}
The first part is by a lexicographic combination of guarded coinduction to prove any instances of the $\infty$-marked subtyping judgment prioritized over structural induction on $A$, and the second is a straightforward simultaneous mixed induction and coinduction on the both derivations.
\end{proof}
\begin{figure}
\centering
\begin{empheq}[box=\fbox]{gather*}
\infer[$\leq$Pred]{\Gamma\vdash A\leq B \\ \Gamma,x\To A\mid\phi(x)\vdash\psi(x)}{\Gamma\vdash(A\mid\phi(\cdot))\leq(B\mid\psi(\cdot))}\quad\infer{X=A\\\infty(\Gamma\vdash A\leq B)}{\Gamma\vdash X\leq B}\quad\infer{X=A\\\infty(\Gamma\vdash B\leq A)}{\Gamma\vdash B\leq X}\\
\infer[${\leq}{\to}$]{\Gamma\vdash\A\leq\A' \\ \Gamma,x\To\A\vdash \B(x)\leq\B'(x)}{\Gamma\vdash(x:\A')\to \B(x)\leq(x:\A)\to \B'(x)}\quad\infer[${\leq}{\unit}$]{}{\Gamma\vdash\unit\leq\unit}\quad\infer[${\leq}{\otimes}$]{\Gamma\vdash\A\leq\A' \\ \Gamma,x\To\A\vdash C(x)\leq D(x)}{\Gamma\vdash(x:\A)\otimes C(x)\leq(x:\A')\otimes D(x)}\\
\infer[${\leq}{\oplus}$]{S\subseteq T\\\{\Gamma\vdash A_\ell\leq B_\ell\}_{\ell\in S}}{\Gamma\vdash{\oplus}\{\ell:A_\ell\}_{\ell\in S}\leq{\oplus}\{\ell:B_k\}_{k\in T}}\quad\infer[${\leq}{\&}$]{T\subseteq S\\\{\Gamma\vdash\A_k\leq\B_k\}_{k\in T}}{\Gamma\vdash{\&}\{\ell:\A_\ell\}_{\ell\in S}\leq{\&}\{\ell:\B_k\}_{k\in T}}
\end{empheq}
\caption{Subtyping}
\label{fig:subtyping-rules}
\end{figure}
\subsection{Cut, Snips, and Identity}
The process behind the cut rule forms the core of computation with futures in (DR)SAX: $x\from P(x);Q(x)$ spawns $P$ to perform a non-blocking write to a newly allocated future addressed by $x$ while concurrently proceeding as $Q$, which may perform a blocking read from $x$. We give two forms of the cut rule depending on which premise outputs the cut ``formula'' $\A$. Thinking of $\A$ as a \emph{midcondition}, cuts are analogous to the composition rule in Hoare logic.
\begin{empheq}[box=\fbox]{gather*}
\snipplus\\
\snipneg
\end{empheq}
In the absence of annotations, these rules actually correspond to \emph{snips} in SAX: \emph{analytic cuts} \cite{Smullyan1969jsl} where $\A$ is a subformula of an axiom's principal formula (hence the rule names). % That is, the general cut rule, where $\A$ is an input in both premises, is enabled by annotations. DeYoung et al.~\cite{DeYoung2020fscd} show that while full cut elimination is not possible in the semi-axiomatic sequent calculus, reduction to snips is still possible, culminating in the subformula property required to prove consistency.
%
%\begin{exmp}[Cuts from Snips]
%Annotations in bidirectional natural deduction enable redexes to appear; likewise, annotations in DRSAX enable instances of the (general) cut rule from snips, i.e., where the cut formula is an input in both premises. We show one such derivation below.
%
%\[\infer*[right=snip\textsuperscript{+}]{\Gamma\vdash P(x)\div(x\From\A)\\\infer*[right=AnnoL]{\Gamma,x\To\A\vdash Q(x)\div(z\From\CC)}{\Gamma,x\From\A\vdash(x:\A)\anno Q(x)\div(z\From\CC)}}{\Gamma\vdash x\from P(x);(x:\A)\anno Q(x)\div(z\From\CC)}\]
%\end{exmp}
%
Likewise, the identity rule $y\from x$, which copies the contents of $x$ to $y$, comes in two forms depending on where the principal ``formula'' $\A$ is outputted:
\begin{empheq}[box=\fbox]{gather*}
\idplus\qquad\idneg
\end{empheq}
%
%\begin{exmp}[Identity and Subsumption]
%Subsumption in bidirectional natural deduction corresponds to the (general) identity rule in the sequent calculus, where the principal formulas are inputs on both sides of the sequent. In DRSAX, the same rule arises from a combination of subsumption and the restricted identity rules above. We show one such derivation below.
%
%\[\inferrule*[right={$\leq$\text{L}}]{\inferrule*[right=id]{ }{\Gamma,x\From\A\vdash y\from x\div(y\From\A)} \\ \Gamma\vdash\B\leq\A}{\Gamma,x\To\B\vdash y\from x\div(y\From\A)}\]
%
%\end{exmp}
%
\subsection{Labelled Sums and Lazy Records}\label{sec:drsax-additives}
Let us pick up where we left off with labelled sums and lazy records in Section \ref{sec:drsax-judgments}. Recalling our new judgmental structure explicitly indicating the flow of type information, the right/introduction rule for $\oplus\{\ell:A_\ell\}_{\ell\in S}$ corresponds to the following right axiom.
\[\infer[{$\oplus$R/I, $k\in S$}]{\Gamma\vdash x\From A_k}{\Gamma\vdash y\From\oplus\{\ell:A_\ell\}_{\ell\in S}}\quad\rightsquigarrow\quad\fbox{\oplusR}\]
The process $y.k\cdot x$ writes the tagged value $k\cdot x$ to the future addressed by $y$. As a result, this rule can be viewed as an instance of the assignment rule in Hoare logic. Indeed, the postcondition flows from right to left, becoming a precondition by a suitable substitution. Thus, the type itself need not embed any type refinements. Now, the corresponding left rule is inherited from the sequent calculus:
\begin{flushleft}
$\infer[$\oplus$L]{\{\Gamma,x\To A_k,y\To\oplus\{\ell:A_\ell\}_{\ell\in S}\vdash(z\From C)\}_{k\in S}}{\Gamma,y\To\oplus\{\ell:A_\ell\}_{\ell\in S}\vdash(z\From C)}~\rightsquigarrow$
\end{flushleft}
\begin{flushright}
$\fbox{\oplusL}$
\end{flushright}
The process $\Case y\,\{\ell\cdot x\To P_\ell(x)\}_{\ell\in S}$ proceeds by cases of the tagged address stored in $y$. Adding refinements requires some care: the antecedent $x$ inherits $\phi(k\cdot x)$ to be in \emph{harmony} \cite{Tennant1978} with the right rule above. Thus, $y$ is additionally subject to $y\equiv k\cdot x$ to indicate its relationship to $x$ when typing each case branch $P_k(x)$. This strong form of path sensitivity is necessary to complete the following example.
\begin{exmp}[Negation]
Letting $\bool=\oplus\{\true:\unit,\false:\unit\}$, we can define Boolean negation of $x$, storing the result in $y$, as $P\triangleq\Case x\,\{\true\cdot x'\To y.\false\cdot x',\false\cdot x'\To y.\true\cdot x'\}$. Then, the judgment $x\To\bool\vdash P\div(y\From\bool\mid\phi(x,y))$ is derivable where $\phi(x,y)\triangleq(\is_\true(x)\implies\is_\false(y))\land(\is_\false(x)\implies\is_\true(y))$. In the first branch, for example, the critical point is when $x'\To\bool$ meets $x'\From\bool\mid\phi(x,\false\cdot x')$ via subsumption---the assumption that $x\To\bool\mid x\equiv\true\cdot x'$ is essential.
%The interesting part is when $x'\To\bool,x\To\bool\mid x\equiv\true\cdot x'\vdash\bool\leq\bool\mid\phi(x,y)$
\end{exmp}
To develop the lazy record type $\&\{\ell:\A_\ell\}_{\ell\in S}$, we again refine its ordinary right/introduction rule:
\[\infer[{$\&$}R/I]{\{\Gamma\vdash x\From A_\ell\}_{\ell\in S}}{\Gamma\vdash y\From\&\{\ell:A_\ell\}_{\ell\in S}}~\rightsquigarrow~\fbox{\withR}\]
In this case, the process $\Case y\,\{\ell\cdot x\To P_\ell(x)\}_{\ell\in S}$ writes a \emph{destination-passing} lazy record to $y$, where its $\ell$\textsuperscript{th} projection writes to the $x$ provided. Destination-passing style is key to our asynchronous operational semantics, as a client of $y$ should be able to refer to $x$ even if it has not yet been populated by $P_\ell(x)$. Now, since the succedents of the premises are inputs in the original rule, they must each be handed a new postcondition; hence each $A_\ell$ becomes $A_\ell\mid\phi_\ell(\cdot)$. The twist is that $\psi(y)$ is verified directly by assuming that there is some $y$ subject to $\phi_\ell(y.\ell)$ for each $\ell$ (the ellipsis repeats the record type). As we mentioned in Section \ref{sec:drsax-assertions}, this respects encapsulation of the record by not reifying its contents into the assertion logic. In particular, intensional properties about the record \emph{cannot} be verified if $\phi_\ell$ hides them, i.e., does not mention them. Let us work through a small example to demonstrate.
\begin{exmp}[Record Encapsulation]
Let $P\triangleq\Case y\,\{\fst\cdot x\To z\from z.\pair{};x.\true\cdot z\}$. Then the judgment $\cdot\vdash P\div(y\From\&\{\fst:\bool\mid\is_{\true}\}\mid\is_{\true}(y.\fst))$ is derivable, but $\cdot\vdash P\div(y\From\&\{\fst:\bool\}\mid\is_{\true}(y.\fst))$ is not.
\end{exmp}
Now, we produce the corresponding refined left axiom below, following our preliminary development in Section \ref{sec:drsax-judgments}. We type the process $y.k\cdot x$, which allows the $k$\textsuperscript{th} projection of $y$ to populate $x$.
\[\infer[\&E, {$k\in S$}]{\Gamma\vdash y\To\&\{\ell:A_\ell\}_{\ell\in S}}{\Gamma\vdash x\To A_k}\quad\rightsquigarrow\quad\fbox{\withL}\]
While it is tempting to let $x$ also be subject to $x\equiv y.k$ as an analogously strong form of path sensitivity, it would not be type-sound, because that relationship is not made explicit in the typing of the right rule (only in the verification of $\psi(y)$). As a result, we produce the following non-example.
\begin{exmp}[Failure of Swap]
We define the following process $P$ that swaps the components of a record $p$ and stores the result in $q$: $P\triangleq\Case q\,\{\fst\cdot x\To p.\snd\cdot x,\snd\cdot y\To p.\fst\cdot y\}$. Then, the following judgment is \emph{not} derivable:
\[p\To\&\{\fst:A,\snd:B\}\vdash P\div(q\From\{\fst:B\mid\lambda x.\,x\equiv p.\fst,\snd:A\mid\lambda y.\,y\equiv p.\snd\})\]
Given that we are already working in the presence of non-termination, a larger set of effects occurring at negative type \cite{Levy1999tlca} may invalidate this kind of equality anyways.
%
%In particular, we would need support for \emph{ghost addresses} to expose the projected values of $p$ in its type.
\end{exmp}
\subsection{Dependent Types}\label{sec:drsax-deptypes}
We now turn our attention to dependent eager pair $(x:\A)\otimes B(x)$ and function $(x:\A)\to\B(x)$ types. Following our development of the labelled sum type, we convert the right/introduction rule to a right axiom. We once again observe that flowing type information bottom up corresponds to a right-to-left flow.

\noindent
$\infer[{$\otimes$}R/I]{\Gamma\vdash x\From A \\ \Gamma\vdash y\From B(x)}{\Gamma\vdash z\From (x:A)\otimes B(x)}~\rightsquigarrow$
\begin{flushright}
\fbox{\otimesR}
\end{flushright}
The process $z.\pair{x,y}$ writes the pair of $x$ and $y$ to $z$. Sensing that this rule too resembles an instance of the assignment rule in Hoare logic, we can now explain the asymmetry between both conjuncts: while $\phi$ flows to the pair's first component $x$, its second component $y$ inherits $\psi$ by substitution of $\pair{x,y}$ for $z$. Like sums, the left rule ${\otimes}$L follows from harmony with its right axiom. Refer to Figure \ref{fig:proc-ty} for this rule as well as those for the unit type. 
%
%\begin{empheq}[box=\fbox]{gather*}
%\otimesL
%\end{empheq}
%
Path sensitivity enables the following example, which was unavailable for lazy records. Note that we use the usual shorthand $A\otimes B$ for non-dependent pairs.
\begin{exmp}[Swap]
Let $P\triangleq\Case z\,\{\pair{x,y}\To w.\pair{y,x}\}$ be a process that swaps a (non-dependent) pair addressed by $z$ and writes it to $w$. Then, the following judgment is derivable:%
\[z\To A\otimes B\vdash P\div(w\From B\otimes A\mid\forall x,y.\,z\equiv\pair{x,y}\implies w\equiv\pair{y,x})\]
\end{exmp}
Following our approach for lazy records, the refined right rule for the dependent function type is as follows.
%
%\[\infer[$\to$R]{\Gamma,x\To A\vdash P(x,y)\div(y\From B)}{\Gamma\vdash\Case z\,(\pair{x,y}\To P(x,y))\div(z\From A\to B)}\]
%
%
\begin{empheq}[box=\fbox]{gather*}
\toR
\end{empheq}
Like that for lazy records, the process $\Case z\,(\pair{x,y}\To P(x,y))$ writes a destination-passing function to $z$ whose body is $P(x,y)$ where $x$ refers to the argument source and $y$ the result destination. Since $x$ and $y$ take assertions as inputs in the premise, function types include a precondition $\phi$ on $x$ and a postcondition $\psi$ on $y$. As with lazy records, the postcondition $\chi$ on $z$ is verified directly from $\phi$ and $\psi$, leaving the function body encapsulated (again, the ellipsis repeats the function type). Let us look at an example to interrogate abstraction boundaries.
\begin{exmp}[Left Unit of Addition]
\label{ex:rua-i}
Let $\nat=\oplus\{\zero:\unit,\succ:\nat\}$. In Example \ref{ex:rua-ii}, we define $\add(x:\nat,y:\nat,z:\nat\mid x+y\equiv z)$ by induction on $x$, assuming that the uninterpreted function $(+)$ is subject to the appropriate axioms. Now, we package this definition into a process $P$ writing to $w$ with a pair of arguments $p$ and result $z$:
\[P\triangleq\Case w\,(\pair{p,z}\To\Case p\,(\pair{x,y}\To\add(x,y,z)))\]
Then, the following judgment is derivable, which asserts that the left unit of addition is zero by applications to $w$. Note that proving zero as the right unit of addition would require a separate induction on $x$.
\begin{align*}
\cdot\vdash P\div(w\From(p:\nat\otimes\nat)\to(\nat\mid\lambda z.\,\forall x,y.\,p\equiv\pair{x,y}\implies z\equiv x+y)\mid\forall x,y.\,\is_{\zero}(x)\implies w\quo\pair{x,y}\equiv y)
\end{align*}
However, the following judgment is \emph{not} derivable, because the action of addition is hidden by the function's output refined type.
\begin{align*}
\cdot\vdash P\div(w\From(p:\nat\otimes\nat)\to(\nat\mid\lambda z.\,\forall x,y.\,p\equiv\pair{x,y}\implies\is_{\zero}(x)\implies y\equiv z)\mid\forall x,y.\, w\quo\pair{x,y}\equiv x+y)
\end{align*}
Note that these functions are uncurried to avoid having to refine each intermediate function type with the necessary information to prove the final postcondition.
\end{exmp}
Finally, it remains to convert the dependent elimination rule to a left axiom. The former outputs both $A$ and $B$ from top down, forcing the second premise to take $A$ as an input. Likewise in the latter, $\A$ flows to the same side (left) of the sequent, but $\B$ flows to the right.

\noindent
$\infer[{$\to$}E]{\Gamma\vdash z\To(x:A)\to B(x) \\ \Gamma\vdash x\From A}{\Gamma\vdash y\To B(x)}~\rightsquigarrow$%
\begin{flushright}
\fbox{\toL}
\end{flushright}
In this case, the process $z.\pair{x,y}$ passes the argument $x$ and result destination $z$ to the function addressed by $z$. Like lazy records, this left rule also omits strong path sensitivity, recalling our discussion of type soundness and a further integration of effects.
\subsection{Recursion: Assume-Guarantee Reasoning and Recursion Invariants}
Following Lakhani et al.~\cite{Lakhani2022esop}, a definition call outputs its ascribed type signature when the definition body checks against the same signature. To type recursive definitions, we take a mixed inductive-coinductive view of the typing judgment as we did with subtyping with the rule below.
\begin{empheq}[box=\fbox]{gather*}
\call
\end{empheq}
Thus, when a recursive definition is checked against its type signature, recursive calls coinductively produce the typing derivation computed so far, corresponding to \emph{assume-guarantee reasoning} that is standard for both program logics and typing recursion \cite{Brotherston,Derakhshan19arxiv,Jhala2021}. In particular, it seems to be the syntactic reflection of coinductively-defined partial correctness (on which we elaborate in the next section); this connection has been explored by Bell and Chlipala \cite{Bell2016coqpl}. We reproduce Example 22 in \cite{Somayyajula2022fscd} below to show the exact mechanics of this process.
\begin{exmp}[Typing Derivation Circularity]
Recalling $\nat=\oplus\{\zero:\unit,\succ:\nat\}$, the following process definition performs a trivial induction on a natural number and returns unit:
\[\eat(x:\nat,y:\unit)=\Case x\,\{\zero\cdot x'\To y\from x',\succ\cdot x'\To\eat(x',y)\}\]
The typing derivation for its body is as follows (process terms are omitted for space) where $\dagger$ denotes a circular edge and $W$ stands for antecedent weakening.
\begin{gather*}
\infer*[left={\(\oplus\)L}]{\infer*[left={$\leq$L}]{\infer*[left=id\textsuperscript{+}]{ }{\ldots,x'\From\unit\vdash y\From\unit}}{x'\To\unit,x\To\nat\mid x\equiv\zero\cdot x'\vdash y\From\unit}\\\infer*[left={$\leq$R,$\leq$L,W}]{\infer*[left=call]{\infty(x'\To\nat\vdash y\From\unit)\quad\dagger}{x'\From\nat\vdash y\To\unit}}{x'\To\nat,x\To\nat\mid x\equiv\succ\cdot x'\vdash y\From\unit}}{x\To\nat\vdash y\From\unit\quad\dagger}
\end{gather*}
\end{exmp}
Notice above that \emph{because} definition calls output, there is a mandatory change of phase between checking the body and the call, which may strengthen the postcondition and weaken the preconditions. This is analogous to checking whether the \emph{loop invariant} implies the postcondition in the loop rule in Hoare logic. As demonstrated in the examples below, our formulation treats induction, coinduction, and mixed induction and coinduction uniformly.
\begin{exmp}[Addition]
\label{ex:rua-ii}
Recall from Example \ref{ex:rua-i} that we will define addition where $\nat=\oplus\{\zero:\unit,\succ:\nat\}$ and $(+)$ is subject to the axioms $\forall x,y.\,\zero\cdot x+y\equiv y$ and $\forall x,y.\,\succ\cdot x+y\equiv\succ\cdot(x+y)$.
\begin{align*}
\add(x:\nat,y:\nat,z:\nat&\mid z\equiv x+y)=\\
&\Case x\,\{\zero\cdot x'\To z\from y,\succ\cdot x'\To z'\from\add(x',y,z');z.\succ\cdot z'\}
\end{align*}
Of significance is checking the snip in the $\succ$ branch: $\add(x',y,z')$ flows $z'\To\nat\mid z'=x'+y$ (the induction hypothesis from typing circularity) to the right but $z.\succ\cdot z'$ flows $z'\From\nat\mid\succ\cdot z'\equiv x+y$ to the left (the induction step). Path sensitivity gives us $x\To\nat\mid x\equiv\succ\cdot x'$, resolving the tension by subsumption.
\end{exmp}
\begin{exmp}[Nonzero Lazy Streams]
In type refinement systems, (co)inductive invariants are typically folded into refinements \cite{Mastorou2022haskell} in lieu of being defined as separate (co)predicates \cite{Momigliano2003types,Leino2014fm}. For example, if $\str=\&\{\head:\nat,\tail:\str\}$ classifies natural number streams, then $\sstr=\&\{\head:\nat\mid\is_\succ,\tail:\sstr\}$ classifies those that are pointwise nonzero. We can check the following definition, which shows that a certain increasing stream starting from a nonzero number is pointwise nonzero.
\[\up(x:\nat\mid\is_{\succ},y:\sstr)=\Case y\,\{\head\cdot h\To h\from x,\tail\cdot t\To x'\from x'.\succ\cdot x;\up(x',t)\}\]

Considering the body of $\up$ as the coinduction step, the coinduction hypothesis (i.e., that the tail is pointwise nonzero) is implicitly part of $\sstr$ as outputted by the recursive call to populate $t$.
\end{exmp}
\begin{exmp}[Left-Fair Streams]
We can extend the technique from the previous example to operate on mixed inductive-coinductive data structures. For example, consider the type below of \emph{left-fair streams} \cite{Basold2018} where, assuming termination, consecutive elements of type $\A$ are interspersed with finitely many timeout ($\later$) labels.
\[\lfair=\oplus\{\now:\&\{\head:\A,\tail:\lfair\},\later:\lfair\}\]
We will define a \emph{projection} operation that is guaranteed to clear these labels, producing the underlying stream. For the sake of this example, we define (later-less) streams as the following recursively refined record:
\[\str=\&\{\fst:\oplus\{\now:\&\{\head:\A,\tail:\str\},\later:\str\}\mid\is_{\now}\}\]
In the definition below, the desired invariant is implicitly checked by coinduction to construct the stream prioritized over induction to vacate the $\later$ labels.
\begin{align*}
\proj(x:\lfair,y:\str)=&\Case y\,\{\fst\cdot l\To\\
&\Case x\,\{\now\cdot s\To s'\from\Case s'\,\{&&\head\cdot h\To s.\head\cdot h,\\
& &&\tail\cdot y'\To x'\from s.\tail\cdot x';\proj(x',y')\};l.\now\cdot s',\\
&~~\quad\qquad\later\cdot x'\To\proj(x',y)\}\}
\end{align*}
Since this process definition is complex, we turn the reader's attention specifically to $l.\now\cdot s'$ which outputs $s'\From\&\{\ldots\}\mid\is_\now(\now\cdot s')$. Checking $s'$ using $\&$R on the left-hand side of the cut then directly verifies $\is_\now(\now\cdot s')$, as desired. Finally, the second recursive call trivially preserves the invariant.
\end{exmp}
We finish this subsection by commenting on the seemingly dangerous interaction between non-termination and type soundness.
\begin{rem}[Non-Termination]
In a cut, non-termination on the left allows the assumption of $\bot$ on the right. Thus, unrestricted lazy evaluation would be incompatible with type soundness, because an unused non-terminating computation can be discarded, exposing a potentially unsafe computation checked against $\bot$ \cite{Vazou2014icfp}. This is not an issue in DRSAX, as the futures-based (as opposed to \emph{speculations}-based \cite[Chapter 38]{Harper2016PFPL}) operational semantics defined in the next section does not discard computations.
\end{rem}
\subsection{Summary}
The process typing rules reviewed in the previous subsections are collected into Figure \ref{fig:proc-ty}.
\begin{figure}
\centering
\begin{empheq}[box=\fbox]{gather*}
\leqR\qquad\leqL\\
\annoR\qquad\annoL\\
\HRule\\
\snipplus\\
\snipneg\\
\idplus\qquad\idneg\\
\HRule\\
\unitR\qquad\unitL\\
\otimesR\\
\otimesL\\
\toR\\
\toL\\
\oplusR\\
\oplusL\\
\withR\\
\withL\\
\HRule\\
\call
\end{empheq}
\caption{Process Typing}
\label{fig:proc-ty}
\end{figure}
\section{Operational Semantics, Type Soundness, and Observable Partial Correctness}\label{sec:semantics}
In this section, we define typing and reduction for \emph{configurations} of DRSAX processes and the future cells with which they communicate. Then, we prove syntactic type soundness. As we alluded in the introduction, this entails \emph{observable} partial correctness, in which hereditarily non-encapsulated sub-configurations (of \emph{purely} positive type) satisfy their associated postconditions directly.%
\subsection{Configuration Reduction and Typing}
\begin{defn}[Configuration]
Configurations are multisets of process and cell objects defined by the following grammar.
\begin{align*}
\C&:=\cdot&&\text{empty configuration}\\
&\mid\,\proc{a}{P}&&\text{process}~P~\text{writing to cell addressed by}~a\\
&\mid\,\cell{a}{S}&&\text{persistent cell addressed by}~a~\text{with contents}~S:=V\mid K\\
&\mid\,\C,\C&&\text{join of two configurations}
\end{align*}
That is, the join and empty rules form a commutative monoid. A configuration $\F$ is \emph{final} when it only consists of cells.
\end{defn}
\begin{figure}
\centering
\begin{minipage}{.4\textwidth}
\begin{center}
\begin{align*}
\cell{a}{W},\proc{b}{(b\from a)}&\mapsto\,\cell{b}{W}\\
\proc{c}{(x\from P(x);Q(x))}&\mapsto\\&\hspace{-8em}\proc{a}{(P(a))},\proc{c}{(Q(a))}~\text{where}~a~\text{is fresh}\\
\cell{a}{K},\proc{c}{(a.V)}&\mapsto\proc{c}{(V\triangleright K)}\\
\cell{a}{V},\proc{c}{(\Case a\,K)}&\mapsto\proc{c}{(V\triangleright K)}\\
\proc{a}{(a\from f~\overline{b})}&\mapsto\proc{a}{(P_f(\overline{b},a))}\\
\proc{a}{(a.V)}&\mapsto\,\cell{a}{V}\\
\proc{a}{(\Case a\,K)}&\mapsto\,\cell{a}{K}
\end{align*}
\end{center}
\end{minipage}
~\vline~
\begin{minipage}{.3\textwidth}
\begin{center}
\begin{align*}
\pair{}\triangleright\pair{}\To P&=P\\
\pair{a,b}\triangleright(\pair{x,y}\To P(x,y))&=P(a,b)\\
k\cdot a\triangleright\{\ell\cdot x\To P_\ell(x)\}_{\ell\in S}&=P_k(a)
\end{align*}
\end{center}
\end{minipage}
\caption{Configuration Reduction}
\label{fig:opsem}
\end{figure}
\emph{Configuration reduction} $(\mapsto)$ is defined by \emph{multiset rewriting rules} \cite{Cervesato09} in Figure \ref{fig:opsem}, which replace any subset of a configuration matching the left-hand side with the right-hand side. $!$ indicates objects that persist across reductions. Now, because bidirectional typing would complicate configuration typing, we first define the corresponding \emph{non-bidirectional} process typing below.
\begin{defn}[Non-bidirectional Typing]
Let $x:\A,\ldots,y:\B(x,\ldots)\vdash P\div(z:\A(x,\ldots,y))$ be generated by rules identical to those in Figure \ref{fig:proc-ty}, but with $\To$ and $\From$ replaced by $(:)$ and \textsc{AnnoL/R} removed.
\end{defn}
As usual, we must verify that the bidirectional process typing is sound and complete with respect to the above.
\begin{lem}[Soundness and Completeness of Bidirectional Typing] Let $\abs{P}$ erase type annotations in $P$ and $\abs{J}$ turn $\To$ and $\From$ to $(:)$. Extending $\abs{\cdot}$ to $\Gamma$ in the obvious way, $\Gamma\vdash P\div J$ iff $\abs{\Gamma}\vdash\abs{P}\div\abs{J}$. 
\end{lem}
\begin{proof}
Both are a routine mixed induction and coinduction on the typing derivation; going forwards erases \textsc{AnnoL/R} and going backwards essentially inserts \textsc{AnnoL/R} as dictated by $P$.
\end{proof}
From now, $\Gamma$ and $\Delta$ refer to contexts associating runtime addresses to refined types. Thus, as a slight abuse of notation, we allow runtime addresses to stand in place of address variables in process typing. Finally, the \emph{configuration typing} judgment $\Gamma\vdash\C\div\Delta$ is inductively generated by the rules in Figure \ref{tab:config-ty}, which types the objects in $\C$ where sources are in $\Gamma$ and destinations in $\Delta$. The rules are designed to admit the following conveniences.
\begin{rem}[Proof Principles]
The theorems in the next subsection use the following proof principles.
\begin{itemize}
\item\emph{Right-to-left induction}: a configuration typing derivation $D$ can be viewed as a list of process typing derivations where readers of an address appear to the right of its writer. Thus, induction on $D$ isolates the rightmost derivation and applies the induction hypothesis to the sub-configuration on the left.
\item\emph{Inversion modulo subtyping}: following \cite{Das2018icfp}, the \textsc{Proc} rule contains subtyping ``slack'' premises on both sides of the sequent. Note that the premise $\Delta\leq\Gamma$ is defined by viewing $\Delta$ and $\Gamma$ as iterated dependent pair types. Thus, for inversion on the typing derivations for processes writing to and reading from the same address, it suffices to only consider the case where they end in right and left rule instances for the same type constructor, respectively. To see why, we first restrict our attention to process typing derivations ending in a non-subsumption rule instance, because terminal instances of subsumption can be absorbed into the ``slack'' using transitivity of subtyping (Lemma \ref{lem:subtyping}). Then, we observe that writers and readers ascribe types $A\leq B\leq\ldots$ to said address. Yet, if $A\leq B$, then $A$ and $B$ have the same head constructor modulo unfolding of recursive types.
\end{itemize}
\end{rem}
%Notice then that the typing rules preserve the invariant $\Sigma\subseteq\Delta$ thanks to the persistence of memory cells.
%
\begin{figure}
\centering
\begin{empheq}[box=\fbox]{gather*}
\infer[proc]{\Delta\leq\Gamma\\\Gamma\vdash P\div(a:\A)\\\Gamma\vdash\A\leq\B}{\Delta\vdash\proc{a}{P}\div(\Delta,a:\B)}\quad\infer[cellV]{\Gamma\vdash \proc{a}{(a.V)}\div\Delta}{\Gamma\vdash\cell{a}{V}\div\Delta}\quad\infer[cellK]{\Gamma\vdash\proc{a}{(\Case a\,K)}\div\Delta}{\Gamma\vdash\cell{a}{K}\div\Delta}\\
\infer[empty]{}{\Gamma\vdash\cdot\div\Gamma}\quad\infer[join]{\Gamma\vdash\C\div\Gamma'\\\Gamma'\vdash\C'\div\Delta}{\Gamma\vdash\C,\C'\div\Delta}
\end{empheq}
\caption{Configuration Typing}
\label{tab:config-ty}
\end{figure}
\subsection{Syntactic Type Soundness and Observable Partial Correctness}
Now that we have reviewed configuration reduction and typing, we prove syntactic type soundness by a standard appeal to progress and preservation. Then, we define and prove observable partial correctness.
\begin{theorem}[Progress]
If $\cdot\vdash\C\div\Delta$ then either $\C$ is final or $\C\mapsto\C'$ for some $\C'$.
\end{theorem}
\begin{proof}
By right-to-left induction on the configuration typing derivation.
\begin{enumerate}
\item If $\C=\C_1,\cell{a}{S}$, then by the induction hypothesis, either $\C_1$ is final, in which case $\C$ is final, or $\C_1\mapsto\C_1'$, in which case $\C\mapsto\C_1',\cell{a}{S}$.
\item If $\C=\C_1,\proc{c}{P}$, then by the induction hypothesis, either $\C_1\mapsto\C_1'$, in which case $\C\mapsto\C_1',\proc{a}{P}$. Otherwise, $\C_1$ is final. If $P$ is a cut, definition call, or writes, then $\C$ steps by $P$ alone. Otherwise, inversion modulo subtyping on the appropriate subderivation in that for $\C_1$ reveals a cell of the right shape that $P$ reads from, letting $\C$ step.
\end{enumerate}
\end{proof}
\begin{theorem}[Preservation]
If $\Gamma\vdash\C\div\Delta$ and $\C\mapsto\C'$, then $\Gamma\vdash\C\div\Delta'$ for some $\Delta'\supseteq\Delta$.
\end{theorem}
\begin{proof}
We proceed by induction on the reduction step and then by inversion modulo subtyping on the typing derivation $D$. The cases where a single process steps---cuts, definition calls, and writes---are straightforward. The identity rule and projection/application of a continuation are also straightforward by copying the derivation of the object read to that of the destination. However, pattern matching on a value is non-trivial due to path sensitivity. For example, when $\oplus$R meets $\oplus$L at address $b$ subject to $\phi(b)$, the $k$\textsuperscript{th} premise of $\oplus$L requires the type of $b$ to be strengthened with the equality $b\equiv k\cdot a$ where $a$ is somewhere to the left in $D$. In updating the derivation for $b$ locally, $a$ would be flowed $\phi(k\cdot a)\land k\cdot a\equiv k\cdot a$, which is subsumed by $\phi(k\cdot a)$ via $\leq$L. Thus, the readers of $a$ see the same type ascription as before. To ensure that all readers of $b$ except for the scrutinized instance of $\oplus$L see the same type ascription as before, we inductively update their left ``slack'' premises noting that $\phi(b)\land b\equiv k\cdot a$ implies $\phi(b)$. %Without loss of generality, we proceed with the other cases when $\C$ only contains the objects distinguished by reduction.
\end{proof}
For an alternate proof strategy of type preservation that grapples with this strong form of path sensitivity in a functional setting, see \cite[Lemme 13.8.7 and Théorème 13.8.8]{RegisGianas2007thesis}. Now, to prove observable partial correctness, we follow DeYoung et.~al.~\cite{DeYoung2020fscd} and refer to addresses occurring in values as \emph{observable} with all else being \emph{hidden}. As a result, final configurations of purely positive type, whose only constituents are value cells, only contain observable addresses.
\begin{lem}[{Final Configurations of Purely Positive Type \cite[Corollary 12]{DeYoung2020fscd}}]
\label{lem:finalpos}
\emph{Purely} positive refined types $\A^{++}$ are those that only contain positive type constructors. Extending this definition to $\Gamma$ in the obvious way, if $\cdot\vdash\F\div\Gamma^{++}$, then $\F$ only contains objects of the form $\cell{a}{V}$ (whose addresses are observable).
\end{lem}
\begin{proof}
By right-to-left induction on the configuration typing derivation, inversion modulo subtyping on the process typing derivation for each cell reveals a value.
\end{proof}
By taking care of the indirection that observable addresses introduce, we can determine when such a final configuration \emph{satisfies} all of its postconditions.
\begin{theorem}[Observable Satisfaction]
$\F\vDash\Gamma^{++}$ is inductively generated by the rules below.
\[\infer{ }{\cdot\vDash\cdot}\qquad\infer{\F\vDash\Gamma^{++}\\\Gamma\vdash\phi(V)}{\F,\cell{a}{V}\vDash\Gamma^{++},a:A\mid\phi(\cdot)}\]
Now, if $\cdot\vdash\F\div\Gamma^{++}$, then $\F\vDash\Gamma^{++}$.
\end{theorem}
\begin{proof}
By right-to-left induction on the configuration typing derivation, we have $\F=\F_1,\cell{a}{V}$ and $\Gamma^{++}=\Gamma_1^{++},a:A\mid\phi(\cdot)$ where $\F_1\vDash\Gamma_1^{++}$. Thus, it suffices to prove $\Gamma_1^{++}\vdash\phi(V)$. By Lemma \ref{lem:finalpos}, said derivation ends in an instance of \textsc{cellV} exposing a process typing derivation. By inversion modulo subtyping, it suffices to only consider right axioms, in which case $\phi(V)$ is already assumed in $\Gamma_1^{++}$ or is proved directly. For example, $\oplus$R assumes $\phi(k\cdot b)$ to type $a.k\cdot b$, whereas $\unit$R proves $\Gamma_1^{++}\vdash\phi(\pair{})$ for $a.\pair{}$.
\end{proof}
Thus, a well-typed configuration is observably partially correct---it either does not terminate or terminates at a final configuration where all of its purely positive subconfigurations observably satisfy their associated postconditions. We formalize this by the following corollary, which combines a coinductive characterization of type soundness \cite{Leroy2006esop} and partial correctness \cite{Clarke1977focs,Goguen1999mscs,Moore2018esop}.
\begin{cor}[Type Soundness and Observable Partial Correctness]
Let $\F$ be $\Gamma$-safe iff for all $\Gamma^{++}\subseteq\Gamma$, there exists $\F'\subseteq\F$ such that $\F'\vDash\Gamma^{++}$. Then, let $\C$ be $\Gamma$-safe iff, coinductively, $\C\mapsto\C'$ and $\C'$ is $\Gamma$-safe. Thus, if $\cdot\vdash\C\div\Gamma$, then $\C$ is $\Gamma$-safe.
\end{cor}
We finish by commenting on the generality of our partial correctness result---because hidden addresses can be made observable by projecting or applying the continuations that hide them, we do not lose power by restricting our attention to observability.% Moreover, negative right rules directly assert the given postconditions with additional assumptions.
\section{Related Work}
We view DRSAX on a spectrum between languages that \emph{model} concurrency and/or parallelism without native support for them at one end and process calculi with dependent (session) types of varying expressivity at the other. Before we elaborate on this dichotomy, we note that our treatment of codata seems to be related to logical approaches to object encapsulation in the presence of mutable state \cite{Hoare1972acta,Hoare2002,OHearn2009toplas}. Moreover, refer to \cite{Basold2016lics,Basold2018,Basold2019jlc} for reasoning about terminating mixed inductive-coinductive programs.
\subsection{Language-Based Verification, Concurrency, and Parallelism}\label{subsec:relwork-first}
Projects like SteelCore \cite{Swamy2020icfp} and FCSL \cite{Nanevski2014esop} implement a variation of \emph{concurrent separation logic} \cite{OHearn2004concur,Jung2018jfp} in a metalanguage---in these cases, F\textsuperscript{$\ast$} or Coq---from which various shared memory and message-passing constructs can be modeled. Similar efforts that do not use separation logic include that in Dafny \cite{Leino2018} and Why3 \cite{Santos2015ice}. Our interest is ``one level up''---determining a core language that could, in theory, be embedded in the languages discussed via the constructs that they model, intersecting with our discussion of embedded session types below. One exception to this thread is Liquid Effects \cite{Kawaguchi2012pldi}, in which dependent type refinements are retrofitted directly onto a parallel dialect of C.
\subsection{Dependent and Embedded Session Types}\label{subsec:relwork-snd}
Toninho et al.~\cite{Toninho2011ppdp} initiated the line of work on dependent session types by presenting a session-typed process calculus in Curry-Howard correspondence with first-order intuitionistic linear logic over a domain of non-linear proof terms. In particular, proof terms are not allowed to refer to the channels with which processes communicate in the linear layer. In their retrospective paper ten years later \cite{Toninho2021ppdp}, they note that many subsequent developments \cite{Toninho2017jlamp,Thiemann2020popl,Das20concur} have similar restrictions precisely because non-linear dependence on linear objects is problematic. As somewhat of an exception to the rule, Toninho and Yoshida \cite{Toninho2018fossacs} allow proof terms to depend on quoted processes by way of a \emph{contextual monad} \cite{Toninho2013esop}, related to that of dependent linear/non-linear logic \cite{Krishnaswami15}. The relaxation of the restriction on type dependency comes at the cost of process/term-level duplication, since functional terms can be embedded faithfully into processes---DRSAX need not make this distinction.% That being said, making adjoint (substructural) SAX \cite{Pruiksma2021jlamp} dependent thus remains an open question.

Another line of work seeks to embed session type systems into existing dependent type theories, allowing meta-level reasoning about processes and the exploitation of existing language infrastructure \cite{Brady2010,Wu2017corr,Hughes2019places,Effpi,Hinrichsen2020popl,Marshall2022places}. Embedded implementation is certainly not opposed by DRSAX nor the line of work above, but moving the burden of proof to the meta level requires explicit reasoning about the typing and operational semantics of programs to an extent determined by the embedding depth.
\section{Conclusion and Future Work}
In this paper, we have developed DRSAX, a sound integration of expressive dependent type refinements into SAX, a futures-based process calculus, by adhering to its proof-theoretic discipline. The distinction between data and codata is navigated through the design of the language as well as the metatheory, which begins with typing rules respecting codata encapsulation and culminates in observable partial correctness as a result of type soundness. Moreover, our mixed inductive-coinductive view of (sub)typing gives a uniform treatment of induction, coinduction, and mixed induction and coinduction within the language. There are multiple avenues of future work:
\begin{enumerate}
\item\emph{Types}: we are interested in extending the type structure of DRSAX primarily by abstraction both over types \cite{Das2022toplas} and refinements \cite{Vazou2013esop}.
\item\emph{Effects}: whether there is a proof-theoretic interpretation of various concurrent effects is still an open question. Non-mutable memory reuse can be interpreted with snips \cite{DeYoung2022mfps}, thus raising the question of how mutability could be introduced. In the setting of session types, \emph{hypersequents} have been used to introduce races in linear logic \cite{Kokke2019coordination}.
\item\emph{Implementation}: the presence of quantifiers in the assertion logic and our mixed inductive-coinductive view of (sub)typing jeopardizes decidable typechecking. With an eye towards implementation, we aim to investigate various quantifier instantiation strategies as well as a restriction to \emph{circular} \cite{Dagnino} (sub)typing derivations which are finitely-representable and thus may admit terminating search \cite{Das20concur}.
\end{enumerate}
%\paragraph{Acknowledgements}
%We would like to thank Henry DeYoung, Klaas Pruiksma, Viktor Kun\^cak, Andrew Myers, and Ranjit Jhala for discussion regarding the contents of this paper.
\bibliographystyle{entics}
\bibliography{refs}

\begin{thebibliography}{10}
\providecommand{\url}[1]{\texttt{#1}}
\providecommand{\urlprefix}{ }
\providecommand{\eprint}[2][]{\url{#2}}

\bibitem{Abel2013popl}
Abel, A., B.~Pientka, D.~Thibodeau and A.~Setzer, \emph{{Copatterns:
  Programming Infinite Structures by Observations}}, in: \emph{Proceedings of
  the 40th Annual ACM SIGPLAN-SIGACT Symposium on Principles of Programming
  Languages}, POPL '13, page 27–38, Association for Computing Machinery, New
  York, NY, USA (2013), ISBN 9781450318327.
\newline\urlprefix\url{https://doi.org/10.1145/2429069.2429075}

\bibitem{Aspinall1996lics}
Aspinall, D. and A.~Compagnoni, \emph{Subtyping dependent types}, in:
  \emph{Proceedings 11th Annual IEEE Symposium on Logic in Computer Science},
  pages 86--97 (1996).
\newline\urlprefix\url{https://doi.org/10.1109/LICS.1996.561307}

\bibitem{Basold2018}
Basold, H., \emph{{Mixed Inductive-Coinductive Reasoning: Types, Programs and
  Logic}}, Ph.D. thesis, Radboud University (2018). Available online at \url{http://cs.ru.nl/~hbasold/thesis/Thesis.pdf}

\bibitem{Basold2016lics}
Basold, H. and H.~Geuvers, \emph{{Type Theory Based on Dependent Inductive and
  Coinductive Types}}, in: \emph{Proceedings of the 31st Annual ACM/IEEE
  Symposium on Logic in Computer Science}, LICS '16, page 327–336,
  Association for Computing Machinery, New York, NY, USA (2016), ISBN
  9781450343916.
\newline\urlprefix\url{https://doi.org/10.1145/2933575.2934514}

\bibitem{Basold2019jlc}
Basold, H. and H.~H. Hansen, \emph{{Well-definedness and observational
  equivalence for inductive–coinductive programs}}, Journal of Logic and
  Computation \textbf{29}, pages 419--468 (2019), ISSN 0955-792X.
  \eprint{https://academic.oup.com/logcom/article-pdf/29/4/419/28917451/exv091.pdf}.
\newline\urlprefix\url{https://doi.org/10.1093/logcom/exv091}

\bibitem{Bell2016coqpl}
Bell, C.~J. and A.~Chlipala, \emph{{A Coinduction Proof Rule for Hoare
  Doubles}}, in: \emph{The 2nd International Workshop on Coq for PL (CoqPL
  2016)}, 43rd Annual ACM SIGPLAN-SIGACT Symposium on Principles of Programming
  Languages (POPL 2016) (2016). Available online at \url{http://people.csail.mit.edu/cj/docs/hdcoind.pdf}

\bibitem{Brady2010}
Brady, E. and K.~Hammond, \emph{{Correct-by-Construction Concurrency: Using
  Dependent Types to Verify Implementations of Effectful Resource Usage
  Protocols}}, Fundam. Inf. \textbf{102}, page 145–176 (2010), ISSN
  0169-2968.
  \newline\urlprefix\url{https://doi.org/10.3233/FI-2010-303}

\bibitem{Brandt1997tlca}
Brandt, M. and F.~Henglein, \emph{{Coinductive Axiomatization of Recursive Type
  Equality and Subtyping}}, in: P.~de~Groote and J.~Roger~Hindley, editors,
  \emph{Typed Lambda Calculi and Applications}, pages 63--81, Springer Berlin
  Heidelberg, Berlin, Heidelberg (1997), ISBN 978-3-540-68438-1.
    \newline\urlprefix\url{https://doi.org/10.1007/3-540-62688-3_29}

\bibitem{Brotherston}
Brotherston, J., \emph{{Cyclic Proofs for First-Order Logic with Inductive
  Definitions}}, in: B.~Beckert, editor, \emph{Automated Reasoning with
  Analytic Tableaux and Related Methods}, pages 78--92, Springer Berlin
  Heidelberg, Berlin, Heidelberg (2005), ISBN 978-3-540-31822-4.
  \newline\urlprefix\url{https://doi.org/10.1007/11554554_8}

\bibitem{Cervesato09}
Cervesato, I. and A.~Scedrov, \emph{{Relating State-Based and Process-Based
  Concurrency through Linear Logic}}, Information and Computation \textbf{207},
  pages 1044--1077 (2009), ISSN 0890-5401. Special issue: 13th Workshop on
  Logic, Language, Information and Computation (WoLLIC 2006).
    \newline\urlprefix\url{https://doi.org/10.1016/j.ic.2008.11.006}

\bibitem{Chlipala2005tldi}
Chlipala, A., L.~Petersen and R.~Harper, \emph{Strict bidirectional type
  checking}, in: \emph{Proceedings of the 2005 ACM SIGPLAN International
  Workshop on Types in Languages Design and Implementation}, TLDI '05, page
  71–78, Association for Computing Machinery, New York, NY, USA (2005), ISBN
  1581139993.
\newline\urlprefix\url{https://doi.org/10.1145/1040294.1040301}

\bibitem{Clarke1977focs}
Clarke, E.~M., \emph{{Program Invariants as Fixed Points (Preliminary
  Reports)}}, in: \emph{18th Annual Symposium on Foundations of Computer
  Science, Providence, Rhode Island, USA, 31 October - 1 November 1977}, pages
  18--29, {IEEE} Computer Society (1977).
\newline\urlprefix\url{https://doi.org/10.1109/SFCS.1977.25}

\bibitem{Dagnino}
Dagnino, F., \emph{{Foundations of regular coinduction}}, {Logical Methods in
  Computer Science} \textbf{{Volume 17, Issue 4}} (2021).
\newline\urlprefix\url{https://doi.org/10.46298/lmcs-17(4:2)2021}

\bibitem{Danielsson09}
Danielsson, N.~A. and T.~Altenkirch, \emph{{Mixing Induction and Coinduction}}
  (2009), draft. Available online at \url{https://www.cse.chalmers.se/~nad/publications/danielsson-altenkirch-mixing.pdf}

\bibitem{Danielsson2010mpc}
Danielsson, N.~A. and T.~Altenkirch, \emph{{Subtyping, Declaratively}}, in:
  C.~Bolduc, J.~Desharnais and B.~Ktari, editors, \emph{Mathematics of Program
  Construction}, pages 100--118, Springer Berlin Heidelberg, Berlin, Heidelberg
  (2010), ISBN 978-3-642-13321-3.
 \newline\urlprefix\url{https://doi.org/10.1007/978-3-642-13321-3_8}

\bibitem{Das2022toplas}
Das, A., H.~Deyoung, A.~Mordido and F.~Pfenning, \emph{{Nested Session Types}},
  ACM Trans. Program. Lang. Syst. \textbf{44} (2022), ISSN 0164-0925.
\newline\urlprefix\url{https://doi.org/10.1145/3539656}

\bibitem{Das2018icfp}
Das, A., J.~Hoffmann and F.~Pfenning, \emph{{Parallel Complexity Analysis with
  Temporal Session Types}}, Proc. ACM Program. Lang. \textbf{2} (2018).
\newline\urlprefix\url{https://doi.org/10.1145/3236786}

\bibitem{Das20concur}
Das, A. and F.~Pfenning, \emph{{Session Types with Arithmetic Refinements}},
  in: I.~Konnov and L.~Kov{\'a}cs, editors, \emph{31st International Conference
  on Concurrency Theory (CONCUR 2020)}, volume 171 of \emph{Leibniz
  International Proceedings in Informatics (LIPIcs)}, pages 13:1--13:18,
  Schloss Dagstuhl--Leibniz-Zentrum f{\"u}r Informatik, Dagstuhl, Germany
  (2020), ISBN 978-3-95977-160-3, ISSN 1868-8969.
  \newline\urlprefix\url{https://doi.org/10.4230/LIPIcs.CONCUR.2020.13}

\bibitem{DeBruijn1991}
{de Bruijn}, N., \emph{{Telescopic mappings in typed lambda calculus}},
  Information and Computation \textbf{91}, pages 189--204 (1991), ISSN
  0890-5401.
\newline\urlprefix\url{https://doi.org/https://doi.org/10.1016/0890-5401(91)90066-B}

\bibitem{Hughes2019places}
de~Muijnck{-}Hughes, J., E.~C. Brady and W.~Vanderbauwhede,
  \emph{{Value-Dependent Session Design in a Dependently Typed Language}}, in:
  F.~Martins and D.~Orchard, editors, \emph{Proceedings Programming Language
  Approaches to Concurrency- and Communication-cEntric Software, PLACES@ETAPS
  2019, Prague, Czech Republic, 7th April 2019}, volume 291 of \emph{{EPTCS}},
  pages 47--59 (2019). 
\newline\urlprefix\url{https://doi.org/10.4204/EPTCS.291.5}

\bibitem{Derakhshan19arxiv}
Derakhshan, F. and F.~Pfenning, \emph{{Circular Proofs as Session-Typed
  Processes: {A} Local Validity Condition}},   (2019). CoRR \url{https://arxiv.org/abs/1908.01909}


\bibitem{DeYoung2022mfps}
DeYoung, H. and F.~Pfenning, \emph{{Data Layout from a Type-Theoretic
  Perspective}}, {Electronic Notes in Theoretical Informatics and Computer
  Science} \textbf{{Volume 1 - Proceedings of MFPS XXXVIII}} (2023).
\newline\urlprefix\url{https://doi.org/10.46298/entics.10507}

\bibitem{DeYoung2020fscd}
DeYoung, H., F.~Pfenning and K.~Pruiksma, \emph{{Semi-Axiomatic Sequent
  Calculus}}, in: Z.~M. Ariola, editor, \emph{5th International Conference on
  Formal Structures for Computation and Deduction (FSCD 2020)}, volume 167 of
  \emph{Leibniz International Proceedings in Informatics (LIPIcs)}, pages
  29:1--29:22, Schloss Dagstuhl--Leibniz-Zentrum f{\"u}r Informatik, Dagstuhl,
  Germany (2020), ISBN 978-3-95977-155-9, ISSN 1868-8969.
  \newline\urlprefix\url{https://doi.org/10.4230/LIPIcs.FSCD.2020.29}

\bibitem{Dunfield2022acm}
Dunfield, J. and N.~Krishnaswami, \emph{{Bidirectional Typing}}, ACM Comput.
  Surv. \textbf{54} (2021), ISSN 0360-0300.
\newline\urlprefix\url{https://doi.org/10.1145/3450952}

\bibitem{Dunfield2004popl}
Dunfield, J. and F.~Pfenning, \emph{{Tridirectional Typechecking}}, in:
  \emph{Proceedings of the 31st ACM SIGPLAN-SIGACT Symposium on Principles of
  Programming Languages}, POPL '04, page 281–292, Association for Computing
  Machinery, New York, NY, USA (2004), ISBN 158113729X.
\newline\urlprefix\url{https://doi.org/10.1145/964001.964025}

\bibitem{Goguen1999mscs}
Goguen, J.~A. and G.~Malcolm, \emph{{Hidden coinduction: behavioural
  correctness proofs for objects}}, {Mathematical Structures in Computer
  Science} \textbf{9}, page 287–319 (1999).
\newline\urlprefix\url{https://doi.org/10.1017/S0960129599002777}

\bibitem{Halstead85}
Halstead, R.~H., \emph{{MULTILISP: A Language for Concurrent Symbolic
  Computation}}, ACM Trans. Program. Lang. Syst. \textbf{7}, page 501–538
  (1985), ISSN 0164-0925.
 \newline\urlprefix\url{https://doi.org/10.1145/4472.4478}

\bibitem{Harper2016PFPL}
Harper, R., \emph{Practical Foundations for Programming Languages}, Cambridge
  University Press, 2 edition (2016).
\newline\urlprefix\url{https://doi.org/10.1017/CBO9781316576892}

\bibitem{Hinrichsen2020popl}
Hinrichsen, J.~K., J.~Bengtson and R.~Krebbers, \emph{{Actris: Session-Type
  Based Reasoning in Separation Logic}}, Proc. ACM Program. Lang. \textbf{4}
  (2019).
\newline\urlprefix\url{https://doi.org/10.1145/3371074}

\bibitem{Hoare1972acta}
Hoare, C. A.~R., \emph{{Proof of Correctness of Data Representations}}, Acta
  Informatica \textbf{1}, pages 271--281 (1972).
\newline\urlprefix\url{https://doi.org/10.1007/BF00289507}

\bibitem{Hoare2002}
Hoare, C. A.~R., \emph{{Towards a Theory of Parallel Programming}}, pages
  231--244, Springer New York, New York, NY (2002), ISBN 978-1-4757-3472-0.
\newline\urlprefix\url{https://doi.org/10.1007/978-1-4757-3472-0_6}

\bibitem{Jhala2021}
Jhala, R. and N.~Vazou, \emph{{Refinement Types: A Tutorial}}, Found. Trends
  Program. Lang. \textbf{6}, page 159–317 (2021), ISSN 2325-1107.
\newline\urlprefix\url{https://doi.org/10.1561/2500000032}

\bibitem{Jung2018jfp}
Jung, R., R.~Krebbers, J.-H. Jourdan, A.~Bizjak, L.~Birkedal and D.~Dreyer,
  \emph{{Iris from the ground up: A modular foundation for higher-order
  concurrent separation logic}}, Journal of Functional Programming \textbf{28},
  page e20 (2018).
\newline\urlprefix\url{https://doi.org/10.1017/S0956796818000151}

\bibitem{Kawaguchi2012pldi}
Kawaguchi, M., P.~Rondon, A.~Bakst and R.~Jhala, \emph{{Deterministic
  Parallelism via Liquid Effects}}, in: \emph{Proceedings of the 33rd ACM
  SIGPLAN Conference on Programming Language Design and Implementation}, PLDI
  '12, page 45–54, Association for Computing Machinery, New York, NY, USA
  (2012), ISBN 9781450312059.
\newline\urlprefix\url{https://doi.org/10.1145/2254064.2254071}

\bibitem{Kokke2019coordination}
Kokke, W., J.~G. Morris and P.~Wadler, \emph{{Towards Races in Linear Logic}},
  in: H.~Riis~Nielson and E.~Tuosto, editors, \emph{Coordination Models and
  Languages}, pages 37--53, Springer International Publishing, Cham (2019),
  ISBN 978-3-030-22397-7.
  \newline\urlprefix\url{https://doi.org/10.1007/978-3-030-22397-7_3}

\bibitem{Krishnaswami15}
Krishnaswami, N.~R., P.~Pradic and N.~Benton, \emph{{Integrating Linear and
  Dependent Types}}, in: \emph{Proceedings of the 42nd Annual ACM
  SIGPLAN-SIGACT Symposium on Principles of Programming Languages}, POPL '15,
  page 17–30, Association for Computing Machinery, New York, NY, USA (2015),
  ISBN 9781450333009.
  \newline\urlprefix\url{https://doi.org/10.1145/2775051.2676969}

\bibitem{Lakhani2022esop}
Lakhani, Z., A.~Das, H.~DeYoung, A.~Mordido and F.~Pfenning, \emph{{Polarized
  Subtyping}}, in: I.~Sergey, editor, \emph{Programming Languages and Systems -
  31st European Symposium on Programming, {ESOP} 2022, Held as Part of the
  European Joint Conferences on Theory and Practice of Software, {ETAPS} 2022,
  Munich, Germany, April 2-7, 2022, Proceedings}, volume 13240 of \emph{Lecture
  Notes in Computer Science}, pages 431--461, Springer (2022).
\newline\urlprefix\url{https://doi.org/10.1007/978-3-030-99336-8\_16}

\bibitem{Leino2018}
Leino, K. R.~M., \emph{{Modeling Concurrency in Dafny}}, in: J.~P. Bowen,
  Z.~Liu and Z.~Zhang, editors, \emph{Engineering Trustworthy Software
  Systems}, pages 115--142, Springer International Publishing, Cham (2018),
  ISBN 978-3-030-02928-9.
  \newline\urlprefix\url{https://doi.org/10.1007/978-3-030-02928-9_4}

\bibitem{Leino2014fm}
Leino, K. R.~M. and M.~Moskal, \emph{Co-induction simply}, in: C.~Jones,
  P.~Pihlajasaari and J.~Sun, editors, \emph{FM 2014: Formal Methods}, pages
  382--398, Springer International Publishing, Cham (2014), ISBN
  978-3-319-06410-9.
    \newline\urlprefix\url{https://doi.org/10.1007/978-3-319-06410-9_27}

\bibitem{Lengrand2006csl}
Lengrand, S., R.~Dyckhoff and J.~McKinna, \emph{{A Sequent Calculus for Type
  Theory}}, in: Z.~{\'E}sik, editor, \emph{Computer Science Logic}, pages
  441--455, Springer Berlin Heidelberg, Berlin, Heidelberg (2006), ISBN
  978-3-540-45459-5.
  \newline\urlprefix\url{https://doi.org/10.1007/11874683_29}

\bibitem{Leroy2006esop}
Leroy, X., \emph{{Coinductive big-step operational semantics}}, in: \emph{ESOP
  2006: European Symposium on Programming}, number 3924 in LNCS, pages 54--68,
  Springer (2006).
\newline\urlprefix\url{https://doi.org/10.1007/11693024_5}

\bibitem{Levy1999tlca}
Levy, P.~B., \emph{{Call-by-Push-Value: A Subsuming Paradigm}}, in: J.-Y.
  Girard, editor, \emph{Typed Lambda Calculi and Applications}, pages 228--243,
  Springer Berlin Heidelberg, Berlin, Heidelberg (1999), ISBN
  978-3-540-48959-7.
  \newline\urlprefix\url{https://doi.org/10.1007/3-540-48959-2_17}

\bibitem{Marshall2022places}
Marshall, D. and D.~Orchard, \emph{{Replicate, Reuse, Repeat: Capturing
  Non-Linear Communication via Session Types and Graded Modal Types}},
  Electronic Proceedings in Theoretical Computer Science \textbf{356}, pages
  1--11 (2022).
\newline\urlprefix\url{https://doi.org/10.4204/eptcs.356.1}

\bibitem{Mastorou2022haskell}
Mastorou, L., N.~Papaspyrou and N.~Vazou, \emph{{Coinduction Inductively:
  Mechanizing Coinductive Proofs in Liquid Haskell}}, in: \emph{Proceedings of
  the 15th ACM SIGPLAN International Haskell Symposium}, Haskell 2022, page
  1–12, Association for Computing Machinery, New York, NY, USA (2022), ISBN
  9781450394383.
\newline\urlprefix\url{https://doi.org/10.1145/3546189.3549922}

\bibitem{Momigliano2003types}
Momigliano, A. and A.~Tiu, \emph{{Induction and Co-induction in Sequent
  Calculus}}, in: S.~Berardi, M.~Coppo and F.~Damiani, editors, \emph{Types for
  Proofs and Programs}, pages 293--308, Springer Berlin Heidelberg, Berlin,
  Heidelberg (2004), ISBN 978-3-540-24849-1.
  \newline\urlprefix\url{https://doi.org/10.1007/978-3-540-24849-1_19}

\bibitem{Moore2018esop}
Moore, B., L.~Pe{\~n}a and G.~Rosu, \emph{{Program Verification by
  Coinduction}}, in: A.~Ahmed, editor, \emph{Proceedings of the 27th European
  Symposium on Programming (ESOP 2018) held as part of the European Joint
  Conferences on Theory and Practice of Software (ETAPS 2018)}, pages 589--618
  (2018), ISBN 9783319898834.
\newline\urlprefix\url{https://doi.org/10.1007/978-3-319-89884-1_21}

\bibitem{Nanevski2014esop}
Nanevski, A., R.~Ley-Wild, I.~Sergey and G.~A. Delbianco, \emph{{Communicating
  State Transition Systems for Fine-Grained Concurrent Resources}}, in:
  \emph{Proceedings of the 23rd European Symposium on Programming Languages and
  Systems - Volume 8410}, page 290–310, Springer-Verlag, Berlin, Heidelberg
  (2014), ISBN 9783642548321.
\newline\urlprefix\url{https://doi.org/10.1007/978-3-642-54833-8_16}

\bibitem{OHearn2004concur}
O'Hearn, P.~W., \emph{{Resources, Concurrency and Local Reasoning}}, in:
  P.~Gardner and N.~Yoshida, editors, \emph{CONCUR 2004 - Concurrency Theory},
  pages 49--67, Springer Berlin Heidelberg, Berlin, Heidelberg (2004), ISBN
  978-3-540-28644-8.
  \newline\urlprefix\url{https://doi.org/10.1007/978-3-540-28644-8_4}

\bibitem{OHearn2009toplas}
O'Hearn, P.~W., H.~Yang and J.~C. Reynolds, \emph{{Separation and Information
  Hiding}}, ACM Trans. Program. Lang. Syst. \textbf{31} (2009), ISSN 0164-0925.
\newline\urlprefix\url{https://doi.org/10.1145/1498926.1498929}

\bibitem{Oppen1978popl}
Oppen, D.~C., \emph{{Reasoning about Recursively Defined Data Structures}}, in:
  \emph{Proceedings of the 5th ACM SIGACT-SIGPLAN Symposium on Principles of
  Programming Languages}, POPL '78, page 151–157, Association for Computing
  Machinery, New York, NY, USA (1978), ISBN 9781450373487.
\newline\urlprefix\url{https://doi.org/10.1145/512760.512776}

\bibitem{Pierce2000acm}
Pierce, B.~C. and D.~N. Turner, \emph{{Local Type Inference}}, ACM Trans.
  Program. Lang. Syst. \textbf{22}, page 1–44 (2000), ISSN 0164-0925.
\newline\urlprefix\url{https://doi.org/10.1145/345099.345100}

\bibitem{RegisGianas2007thesis}
R{\'e}gis-Gianas, Y., \emph{{Des types aux assertions logiques : preuve
  automatique ou assistée de propriétés sur les programmes fonctionnels}},
  Theses, {Universit{\'e} Paris Diderot} (2007).
\newline\urlprefix\url{https://hal.inria.fr/tel-01238703}

\bibitem{Gianas2008mpc}
R{\'e}gis-Gianas, Y. and F.~Pottier, \emph{{A Hoare Logic for Call-by-Value
  Functional Programs}}, in: P.~Audebaud and C.~Paulin-Mohring, editors,
  \emph{Mathematics of Program Construction}, pages 305--335, Springer Berlin
  Heidelberg, Berlin, Heidelberg (2008), ISBN 978-3-540-70594-9.

\bibitem{Rondon2008pldi}
Rondon, P.~M., M.~Kawaguci and R.~Jhala, \emph{{Liquid Types}}, SIGPLAN Not.
  \textbf{43}, page 159–169 (2008), ISSN 0362-1340.
\newline\urlprefix\url{https://doi.org/10.1145/1379022.1375602}

\bibitem{Rushby1998tse}
Rushby, J., S.~Owre and N.~Shankar, \emph{{Subtypes for specifications:
  predicate subtyping in PVS}}, IEEE Transactions on Software Engineering
  \textbf{24}, pages 709--720 (1998).
\newline\urlprefix\url{https://doi.org/10.1109/32.713327}

\bibitem{Santos2015ice}
Santos, C., F.~Martins and V.~T. Vasconcelos, \emph{{Deductive Verification of
  Parallel Programs Using Why3}}, in: S.~Knight, I.~Lanese, A.~Lluch{-}Lafuente
  and H.~T. Vieira, editors, \emph{Proceedings 8th Interaction and Concurrency
  Experience, {ICE} 2015, Grenoble, France, 4-5th June 2015}, volume 189 of
  \emph{{EPTCS}}, pages 128--142 (2015).
\newline\urlprefix\url{https://doi.org/10.4204/EPTCS.189.11}

\bibitem{Effpi}
Scalas, A., N.~Yoshida and E.~Benussi, \emph{{Effpi: A Toolkit for Verified
  Message-Passing Programs in Dotty}}.
\newline\urlprefix\url{https://doi.org/10.1145/3325968}

\bibitem{Smullyan1969jsl}
Smullyan, R.~M., \emph{Analytic cut}, The Journal of Symbolic Logic
  \textbf{33}, page 560–564 (1969).
\newline\urlprefix\url{https://doi.org/10.2307/2271362}

\bibitem{Somayyajula2022fscd}
Somayyajula, S. and F.~Pfenning, \emph{{Type-Based Termination for Futures}},
  in: \emph{7th International Conference on Formal Structures for Computation
  and Deduction (FSCD 2022)} (2022).

\bibitem{Swamy2020icfp}
Swamy, N., A.~Rastogi, A.~Fromherz, D.~Merigoux, D.~Ahman and G.~Mart\'{\i}nez,
  \emph{{SteelCore: An Extensible Concurrent Separation Logic for Effectful
  Dependently Typed Programs}}, Proc. ACM Program. Lang. \textbf{4} (2020).
\newline\urlprefix\url{https://doi.org/10.1145/3409003}

\bibitem{Tennant1978}
Tennant, N., \emph{{Natural Logic}}, Edinburgh University Press (1978), ISBN: 0852245793

\bibitem{Thiemann2020popl}
Thiemann, P. and V.~T. Vasconcelos, \emph{{Label-Dependent Session Types}},
  Proc. ACM Program. Lang. \textbf{4} (2019).
\newline\urlprefix\url{https://doi.org/10.1145/3371135}

\bibitem{Toninho2011ppdp}
Toninho, B., L.~Caires and F.~Pfenning, \emph{{Dependent Session Types via
  Intuitionistic Linear Type Theory}}, in: \emph{Proceedings of the 13th
  International ACM SIGPLAN Symposium on Principles and Practices of
  Declarative Programming}, PPDP '11, page 161–172, Association for Computing
  Machinery, New York, NY, USA (2011), ISBN 9781450307765.
\newline\urlprefix\url{https://doi.org/10.1145/2003476.2003499}

\bibitem{Toninho2013esop}
Toninho, B., L.~Caires and F.~Pfenning, \emph{{Higher-Order Processes,
  Functions, and Sessions: A Monadic Integration}}, in: M.~Felleisen and
  P.~Gardner, editors, \emph{Programming Languages and Systems}, pages
  350--369, Springer Berlin Heidelberg, Berlin, Heidelberg (2013), ISBN
  978-3-642-37036-6.
  \newline\urlprefix\url{https://doi.org/10.1007/978-3-642-37036-6_20}

\bibitem{Toninho2021ppdp}
Toninho, B., L.~Caires and F.~Pfenning, \emph{{A Decade of Dependent Session
  Types}}, in: \emph{23rd International Symposium on Principles and Practice of
  Declarative Programming}, PPDP 2021, Association for Computing Machinery, New
  York, NY, USA (2021), ISBN 9781450386890.
\newline\urlprefix\url{https://doi.org/10.1145/3479394.3479398}

\bibitem{Toninho2017jlamp}
Toninho, B. and N.~Yoshida, \emph{Certifying data in multiparty session types},
  Journal of Logical and Algebraic Methods in Programming \textbf{90}, pages
  61--83 (2017), ISSN 2352-2208.
\newline\urlprefix\url{https://doi.org/https://doi.org/10.1016/j.jlamp.2016.11.005}

\bibitem{Toninho2018fossacs}
Toninho, B. and N.~Yoshida, \emph{{Depending on Session-Typed Processes}}, in:
  C.~Baier and U.~D. Lago, editors, \emph{Foundations of Software Science and
  Computation Structures - 21st International Conference, {FOSSACS} 2018, Held
  as Part of the European Joint Conferences on Theory and Practice of Software,
  {ETAPS} 2018, Thessaloniki, Greece, April 14-20, 2018, Proceedings}, volume
  10803 of \emph{Lecture Notes in Computer Science}, pages 128--145, Springer
  (2018).
\newline\urlprefix\url{https://doi.org/10.1007/978-3-319-89366-2\_7}

\bibitem{Vazou2022haskell}
Vazou, N. and M.~Greenberg, \emph{{How to Safely Use Extensionality in Liquid
  Haskell}}, in: \emph{Proceedings of the 15th ACM SIGPLAN International
  Haskell Symposium}, Haskell 2022, page 13–26, Association for Computing
  Machinery, New York, NY, USA (2022), ISBN 9781450394383.
\newline\urlprefix\url{https://doi.org/10.1145/3546189.3549919}

\bibitem{Vazou2013esop}
Vazou, N., P.~M. Rondon and R.~Jhala, \emph{{Abstract Refinement Types}}, in:
  M.~Felleisen and P.~Gardner, editors, \emph{Programming Languages and
  Systems}, pages 209--228, Springer Berlin Heidelberg, Berlin, Heidelberg
  (2013), ISBN 978-3-642-37036-6.

\bibitem{Vazou2014icfp}
Vazou, N., E.~L. Seidel, R.~Jhala, D.~Vytiniotis and S.~Peyton-Jones,
  \emph{{Refinement Types for Haskell}}, in: \emph{Proceedings of the 19th ACM
  SIGPLAN International Conference on Functional Programming}, ICFP '14, page
  269–282, Association for Computing Machinery, New York, NY, USA (2014),
  ISBN 9781450328739.
\newline\urlprefix\url{https://doi.org/10.1145/2628136.2628161}

\bibitem{Vazou2018popl}
Vazou, N., A.~Tondwalkar, V.~Choudhury, R.~G. Scott, R.~R. Newton, P.~Wadler
  and R.~Jhala, \emph{{Refinement Reflection: Complete Verification with SMT}},
  Proc. ACM Program. Lang. \textbf{2} (2017).
\newline\urlprefix\url{https://doi.org/10.1145/3158141}

\bibitem{Wu2017corr}
Wu, H. and H.~Xi, \emph{{Dependent Session Types}}, CoRR
  \textbf{abs/1704.07004} (2017). \eprint{1704.07004}.
\newline\urlprefix\url{http://arxiv.org/abs/1704.07004}

\bibitem{Xi1999popl}
Xi, H. and F.~Pfenning, \emph{{Dependent Types in Practical Programming}}, in:
  \emph{Proceedings of the 26th ACM SIGPLAN-SIGACT Symposium on Principles of
  Programming Languages}, POPL '99, page 214–227, Association for Computing
  Machinery, New York, NY, USA (1999), ISBN 1581130953.
\newline\urlprefix\url{https://doi.org/10.1145/292540.292560}

\bibitem{Zeilberger2015types}
Zeilberger, N., \emph{{Balanced polymorphism and linear lambda calculus}}
  (2015). Available online at \url{http://noamz.org/papers/linprin.pdf}.

\end{thebibliography}
\end{document}